\documentclass[aap]{imsart}

\RequirePackage{amsthm,amsmath,amsfonts,amssymb}
\RequirePackage[authoryear]{natbib}
\RequirePackage[colorlinks,citecolor=blue,urlcolor=blue]{hyperref}
\RequirePackage{graphicx}
\usepackage{enumerate}
\usepackage{dsfont}
\usepackage{color}

\startlocaldefs
\theoremstyle{plain}
\newtheorem{thm}{Theorem}[section]
\newtheorem{open}[thm]{Open problem} 

\newtheorem{prop}[thm]{Proposition}
\newtheorem{lemma}[thm]{Lemma}
\newtheorem{cor}[thm]{Corollary}
\theoremstyle{remark}
\newtheorem{example}[thm]{Example}

\newtheorem{remark}[thm]{Remark}
\def\stany{\mathcal{X}}

\def\1c{\mathbb{I}_C(x)}

\def\n0{n_{0}}

\def\Pr{\mathbb{P}}

\def\liM_d{{\lim_{n\to\infty}}}

\def\PM{{\mathsf{P}}}
\def\PiM{{\mathsf{\Pi}}}
\def\dist{{\mathrm{dist\,}}}
\def\Im{{\mathrm{Im\,}}}
\def\vr{{\varrho}}

\endlocaldefs
\begin{document}
\renewcommand{\arraystretch}{1.5}
\begin{frontmatter}

\title{Solidarity of Gibbs Samplers: the~spectral gap}
\runtitle{Solidarity of Gibbs Samplers}

\begin{aug}
\author[A]{ \fnms{Iwona} \snm{Chlebicka} 
\ead[label=e3]{i.chlebicka@mimuw.edu.pl}}
\and
\author[B]{ \fnms{Krzysztof} \snm{{\L}atuszy\'{n}ski} 
\ead[label=e1]{K.G.Latuszynski@warwick.ac.uk}}\\
\and
\author[A]{ \fnms{B{\l}a\.zej} \snm{Miasojedow}
\ead[label=e2]{b.miasojedow@mimuw.edu.pl}} 
\runauthor{I. Chlebicka, K. {\L}atuszy\'nski, B. Miasojedow}

\address[A]{Institute of Applied Mathematics and Mechanics\\ University of Warsaw\\ 
Banacha 2, 02-097 Warszawa, Poland\\
\printead{e3,e2}}
\address[B]{ Department of Statistics\\
University of Warwick\\ CV4 7AL, Coventry, UK\\
\printead{e1}}

\end{aug}

\begin{abstract}
 Gibbs samplers are preeminent Markov chain Monte Carlo algorithms used in computational physics and statistical computing. Yet, their most fundamental properties, such as relations between convergence characteristics of their various versions, are not well understood.
 
 In this paper we prove the solidarity principle of the spectral gap for the Gibbs sampler: if any of the random scan or $d!$ deterministic scans has a~spectral gap then all of them have. Our methods rely on geometric interpretation of the Gibbs samplers as alternating projection algorithms and analysis of the rate of convergence in the von Neumann--Halperin method of cyclic alternating projections. As byproduct of our analysis, we also establish that deterministic scan Gibbs, despite being non-reversible, share many robustness properties with reversible chains, including exponential inequalities and Central Limit Theorems under the same conditions.

In addition, we provide a quantitative result: if the spectral gap of the random scan Gibbs sampler decays with dimension at {  a polynomial rate $\beta$, then the rate is no worse than $2\beta +2$ for any deterministic scan.}

\end{abstract}


\end{frontmatter}

\section{Introduction}

Gibbs samplers are Markov chain Monte Carlo (MCMC) algorithms that stem from statistical mechanics. There, under the name of heat bath or Glauber dynamics, they were introduced to study discrete models, in particular their phase transition (\cite{MR148410}). Following~\cite{gemangeman} and \cite{gelfand1990sampling} Gibbs samplers enjoy tremendous popularity also in mainstream computational sciences. They are used for sampling from intractable high-dimensional probability distributions $\pi$ by iteratively sampling from the conditional distributions of $\pi$ which for many modelling settings, and commonly in the Bayesian paradigm, are more tractable.

The Gibbs samplers are often preferred over other approaches as they tend to be straightforward to implement and easy to use (\cite{MR1210421}). Choosing a scan strategy -- the order in which the individual coordinates of the $d$-dimensional $\pi$ shall be updated from their conditional distributions -- is the key step in setting up a Gibbs sampler. Two common choices are the {\it random scan Gibbs sampler}  ($\PM_{\rm RSG}$) and the  {\it  deterministic scan Gibbs sampler} ($\PM_{\rm DSG}$). The former updates one randomly chosen coordinate at a time while the  latter updates the coordinates one-by-one in a deterministic order. While choosing a particular scan does not represent a problem at the first glance, and any of $\PM_{\rm RSG}$ or the $d!$ possible $\PM_{\rm DSG}$ versions will result in a valid algorithm that has $\pi$ as its invariant distribution, specific guidance on the choice is scarce. In fact, no scan strategy is universally optimal, as shown by \cite{Roberts2015SurprisingCP} and \cite{HeDeSaMiRe}. Furthermore, the efficiency of $\PM_{\rm RSG}$ depends non-trivially on the vector of selection probabilities (\cite{robertssahu97, chimisov2018adapting}). However, the theory of their fundamental properties that would usefully inform about the possible differences between scans is in short supply.

\begin{figure}
\begin{tabular}{cccc} 
  & ${\PM_{\rm RSG}}$ & & ${\PM_{\rm DSG}}$ \\
  spectral gap in $L^2(\pi)$ & $+$ & {$\Longleftrightarrow$} & $+$  \\ 
   &{$\big\Updownarrow$} &  &  \;\;{$\big\Downarrow$}\;\; {${\big\Uparrow}{?}$} \\ 
  geometric ergodicity & $+$& \large{$\substack{\;\;\Longrightarrow_{\phantom{k}} \\{?\Longleftarrow^{\phantom{k}}}}$ }& $+$ \\ 
  &{$\phantom{\bigg\Updownarrow}$} &  &  \;\;{} 
\end{tabular}
\caption{\label{picture} Relations between properties of Gibbs samples due to Theorem~\ref{theo:main} (upper $\Longleftrightarrow$), Corollary~\ref{cor:main} ({\it (i)} is lower $\Longrightarrow$, {\it (ii)} is left $\Uparrow$), together with possible equivalent implications  with question marks (cf. Open Problem~\ref{op} for $\Longleftarrow$). All $\Downarrow$ result from \cite[Theorem~1.3]{kontoyiannis2012geometric}.}
\end{figure}

 The spectral gap property of Markov transition operators is central to Markov chain analysis and closely linked to convergence rates and other stability characteristics of MCMC (\cite{MT2009, MR1952001, kipnis1986central}). With the aim of extending the current understanding of general state space Gibbs samplers, in the present paper we offer a complete characterization of the spectral gap property for different scan strategies by showing the following dichotomy:
 
 \begin{thm}\label{theo:main} Let $\PM_{\rm DSG}$ and  $\PM_{\rm RSG}$ be Gibbs samplers targeting $\pi$. Then exactly one of the following holds:
\begin{enumerate}[(i)]
 \item None of the Gibbs samplers, $\PM_{\rm DSG}$'s or  $\PM_{\rm RSG}$'s,  have a positive  $L^2(\pi)$ spectral gap;
 \item All of the Gibbs samplers, $\PM_{\rm DSG}$'s and  $\PM_{\rm RSG}$'s  have a positive  $L^2(\pi)$ spectral gap.
\end{enumerate}
\end{thm}

Hence, qualitatively, all the deterministic and random scan Gibbs samplers are equally good (see also Figure~\ref{picture}).  We refer to this dichotomy as solidarity principle of the spectral gap. Moreover, since for reversible Markov chains spectral gap is equivalent to geometric ergodicity (GE for short), \cite{MT2009}, the above result implies the following:

\begin{cor}\label{cor:main}
 Assume GE of any random scan Gibbs sampler targeting $\pi$. Then:
\begin{enumerate}[(i)]
 \item all the Gibbs samplers, $\PM_{\rm DSG}$'s and  $\PM_{\rm RSG}$'s, are geometrically ergodic;
 \item all the Gibbs samplers, $\PM_{\rm DSG}$'s and  $\PM_{\rm RSG}$'s,  have a positive $L^2(\pi)$ spectral gap.
\end{enumerate}
\end{cor}
 {Proof of the above corollary is provided in Section~\ref{sec:gaps} right after the proof of Theorem~\ref{theo:main}.}

However, to establish an analogous solidarity principle of GE for Gibbs samplers, one would need to address the following open problem:

 \begin{open}\label{op} Does GE of $\PM_{\rm DSG}$ imply GE of $\PM_{\rm RSG}$?
 \end{open}
 Equivalently, one could also ask if GE of $\PM_{\rm DSG}$ implies its spectral gap. The answer is positive when $\PM_{\rm DSG}$ is reversible, that is in dimension $2$. Otherwise, it is a nontrivial issue discussed in Section~\ref{sec:op}.
 
 Main tools for our proof are provided in Section \ref{sec:gaps} and Theorem~\ref{theo:equiv}. In particular, it implies that despite  $\PM_{\rm DSG}$ being non-reversible, it satisfies an important condition enjoyed by reversible operators: the existence of its spectral gap is equivalent to $\|\PM_{\rm DSG}-\PiM\|<1$. We explore this to establish that deterministic scan Gibbs samplers enjoy 
 many of the key ergodic properties under the same minimal conditions as random scans and other reversible chains: the Hoeffding inequality of~\cite{miasojedowhoeffding} holds for  $\PM_{\rm DSG}$, as does the Central Limit Theorem for functions merely in~$L^2(\pi)$, see Theorem~\ref{thm:clt}, with explicit bounds on the asymptotic variance in terms of the spectral gap. 
 
 The rate of decay of the Markov operator spectral gap with dimension is
 key to understanding how its computational efficiency scales with the problem size. Rapid and torpid mixing, which stand for the polynomial and exponential decay, respectively, have been extensively studied for various algorithms in the settings from statistical physics (\cite{MR1943860,MR2291087}) and Bayesian posterior sampling (\cite{rapid,torpid}). For the Gibbs samplers, there was a longstanding conjecture (c.f. \cite{MR3102552, MR2466937}) that deterministic scan was at most a constant factor slower and at most a logarithmic factor faster than random scan. However, \cite{Roberts2015SurprisingCP} and \cite{HeDeSaMiRe} provided counterexamples for both parts of this claim.

\begin{figure}
\begin{tabular}{cccc}  
  & ${\PM_{\rm RSG}}$ & & ${\PM_{\rm DSG}}$ \\
  rapid mixing & $+$& \large{$\substack{\;\;\Longrightarrow_{\phantom{k}} \\{?\Longleftarrow^{\phantom{k}}}}$ }& $+$ \\ 
\end{tabular}
\caption{\label{picture2} Relations between properties of Gibbs samples due to Corollary~\ref{cor:poly}, together with a possible implication from Open Problem \ref{open:poly} shown with a question mark.}
\end{figure}

 In this spirit, our Corollary~\ref{cor:bounds} provides quantitative estimates on the spectral gaps for $\PM_{\rm DSG}$ and $\PM_{\rm RSG}$ in terms of geometry of the operators. The estimate for $\PM_{\rm RSG}$ is sharp and the result implies (via Corollary~\ref{cor:poly}) 
 {   that rapid mixing of $\PM_{\rm RSG}$ implies rapid mixing of $\PM_{\rm DSG}$. More precisely, if the polynomial rate of decay with dimension of the spectral gap of $\PM_{\rm RSG}$ is $\beta$, than it is no worse than $2\beta +2$ for $\PM_{\rm DSG}$. This fact is illustrated by Figure~\ref{picture2}.}  We propose the opposite implication, that rapid mixing of $\PM_{\rm DSG}$ implies rapid mixing of $\PM_{\rm RSG}$, as Open Problem~\ref{open:poly}.

Our paper contributes to a vast body of related research. 
As the speed of convergence of MCMC is of major importance in theory and applications, the efficiency of swap strategies for the Gibbs samplers were studied in many contexts, however, due to the difficulty of the problem, usually under some simplifying assumptions.
The rates of convergence in the Gaussian and approximately Gaussian case were studied in~\cite{amit1991rates,amitgrenander1991,robertssahu97} and \cite{chimisov2018adapting}. Beyond the Gaussian case, but for two-component samplers, \cite{qin2021}  showed that the deterministic scan converges faster in $L^2(\pi)$ than its random scan counterpart, no matter the selection probabilities in the random scan, and provided a precise relation between the two convergence rates. As shown by \cite{Roberts2015SurprisingCP}, such a result can not hold when $d>2$. Further examples were given in \cite{HeDeSaMiRe} where also the relation between random scans and conveniently modified deterministic scans was studied in discrete models and under further specific assumptions.  {We refer to \cite{JJN} and \cite{TJH} for conditions which 
guarantee geometric ergodicity of all versions of two-component Gibbs samplers simultaneously.} A comparison of two dimensional sweep strategies in terms of asymptotic variance was carried out by \cite{greenwood}  { and further extended in \cite{qin2024analysistwocomponentgibbssamplers}}. In such a case deterministic scan has always smaller asymptotic variance. The topic was studied further for broader family of MCMC algorithms in \cite{andrieu2016biometrika,maire2014,andrieu2021}. However, following the examples of \cite{Roberts2015SurprisingCP} and the discussion in \cite{greenwood}, these asymptotic variance results do not extend beyond $d=2$. Let us note here that this is not surprising: in the study of $L^2(\pi)$ convergence reversibility plays the main  {role, cf. \cite{kontoyiannis2012geometric}}. In the two-component case, the Gibbs samplers $\PM_{\rm DSG}=\PM_1\PM_2$ and $\PM_{\rm RSG}=w_1\PM_1+w_2\PM_2$ are both reversible. This results from the fact that for the projection operators $\PM_1,\PM_2$ we have $(\PM_1\PM_2)^n=(\PM_1\PM_2\PM_1)^n$ where $\PM_1\PM_2\PM_1$ is self-adjoint (and $\PM_{\rm RSG}$ is always reversible in any dimension). Higher dimensional $\PM_{\rm DSG}$, that we analyse in this paper, are not reversible anymore and it is well known that passing to multi-component samplers brings in surprising and counterintuitive phenomena (\cite{Roberts2015SurprisingCP}).

Without assuming (approximate) normality of the target, or restricting to $d=2$, relations between the covariance structure and spectral properties of Gibbs samplers were studied  in~\cite{liuwongkong}, where it is also shown that $\PM_{\rm RSG}$ is positive, i.e. its spectrum is contained in $[0,1]$.
Regarding qualitative properties, Proposition 3.2 of
\cite{roberts1997geometric} implies that if $\PM_{\rm DSG}$ is uniformly ergodic, so is $\PM_{\rm RSG}$, see also Proposition~4.8 of \cite{MR3059204} for a more direct argument.

Beyond the general abstract setting, GE of the Gibbs samplers has been studied for particular data models to establish reliability of Bayesian statistical inference in various statistical frameworks, e.g. by~\cite{hobert1998geometric,johnson2010gibbs,jones2004sufficient,marchev2004geometric,tan2009block}, for general hierarchical models in \cite{MR2387965} and, more recently, also in~\cite{zanella2021,papaspiliopoulos2020} where estimates for multilevel hierarchical Gaussian models are provided. In statistical physics, mixing time for the Ising model was studied via the spectral theory since~\cite{desantis2002,martinelli2004}. For more results on Glauber dynamics of spin glasses we refer to later contributions~\cite{benarous2018,eldan2022,gheissari2019,jagannath2019}.
We point out that the abovementioned contributions cover a wide class of problems, but the general theory still calls for a unified approach.

Describing the Gibbs samplers in the terms of the orthogonal projection operators, the approach that we here follow, was initiated by \cite{amit1991rates} who studied a particular case of convergence of Gibbs samplers for Gaussian targets based on the bounds due to~\cite{smithsolmonwagner}. On the other hand, the complete characterization of the $L^2(\pi)$-convergence rate of the two-dimensional Gibbs samplers was provided in \cite{diaconis2010}. To establish the characterisation,~\cite{diaconis2010} apply the classical results by~\cite{aronszajn} on alternating projection algorithms. Our methods rely on a more refined application of the geometric approach to the theory of operators developed by \cite{badea2012rate} in order to analyze the rate of convergence in the von Neumann--Halperin method of cyclic alternating projections. By Proposition~4.6 in~\cite{badea2012rate} their approach yields more precise estimates than those of~\cite{smithsolmonwagner} or~\cite{deutschhundal}.


 \medskip
 
The rest of the paper is organized as follows. In Section~\ref{sec:prelim} we introduce notation and background material on the the geometric approach. Section~\ref{sec:gaps} is devoted to the spectral properties of the Gibbs samplers. In particular, we present there Theorem~\ref{theo:equiv} together with its proof, followed by the proof of Theorem~\ref{theo:main}. Section~\ref{sec:CTG} provides additional results, including the central limit theorem and the exponential inequality for the $\PM_{\rm DSG}$, while the discussion of Open Problem~\ref{op} is in Section~\ref{sec:op}.

\section{Preliminaries}\label{sec:prelim}

\textbf{Notation. } In the sequel by $I$ we denote the identity operator and by $\Im$ we denote the image (range) of an operator. We set $\mathcal{X} = \stany_1 \times \dots \times \stany_d$ -- a Polish space (metric, separable, and complete) and  $X_n = (X_{n,1}, \dots, X_{n,d})\in \mathcal{X}$. We shall use the shorthand notation
$$
X_{n,-i} \ := \
(X_{n,1}, \dots, X_{n,i-1}, X_{n,i+1},
\dots, X_{n,d})
\, ,
$$
and similarly $\stany_{-i} = \stany_1\times \dots
\times \stany_{i-1}\times \stany_{i+1}\times \dots\times \stany_d$.  Let us consider a~probability measure on $\mathcal{X}$ with a density $\pi$ that will be the target distribution. We denote a  Hilbert space $L^2(\pi)$ as a space of functions $f:\stany\to \mathbb{R}$
with $\pi f^2 =\int_\stany f^2(x)\pi(x)\,dx<\infty$ equipped with an inner product $\langle f,g\rangle=\int_\stany f(x)
g(x)\pi(x)\,dx$.  {We denote the norm in $L^2(\pi)$ by $\Vert \cdot\Vert$.} 
We denote by $\PiM$ a linear operator on $L^2(\pi)$, such that for any function $f\in L^2(\pi)$ we have \begin{equation}
    \label{PiM} \PiM f =\pi(f).
\end{equation} \newline

\noindent \textbf{Gibbs samplers targetting $\pi$. }  
For $Z \sim \pi$ by
$\pi(\cdot| x_{-i})$ we denote the conditional distribution of $Z_{i} \, |
\, Z_{-i}=x_{-i}$.
By a {\it small step} $\PM_i$ we mean the transition operator that replaces $X_n = x_n\in \mathcal{X}$ by $$X_{n+1} \; = \; (X_{n,1}, \dots, X_{n, i-1}, Z_i, X_{n, i+1}, \dots, X_{n, d})\,,$$ where $Z_i \sim \pi(\cdot| x_{-i}).$ Then the {\it  deterministic scan Gibbs sampler} is an algorithm that updates a coordinate using the transition operator  \begin{equation} \label{eqn_gibbs_kernel_DS}
\PM_{\rm DSG} \; = {\; \PM_d\PM_{d-1}\cdots\PM_1}\;= \; \prod_{i=1}^d \PM_i\,,
\end{equation} whereas the {\it random scan Gibbs sampler} uses the transition operator 
\begin{equation} \label{eqn_Gibbs_kernel_RS}
\PM_{\rm RSG} \; = \; \sum_{i=1}^d w_i \PM_i\,,
\end{equation}
where weights $w_i$  satisfy\begin{equation}
    \label{weights}\text{ $w_i>0$ for every $i=1,\dots d\ $ and $\ \sum_{i=1}^d w_i =1\,$.}
\end{equation} A typical choice of weights  in~\eqref{eqn_Gibbs_kernel_RS} is $w_i=\frac{1}{d}$. We recall that unlike the deterministic scan Gibbs sampler, the random scan Gibbs sampler is reversible.\newline

\noindent \textbf{Spectrum and Markov chains. } By $Spec(T)$ we mean a {\it spectrum} of  a bounded linear operator $T$, that is a set of numbers $\lambda$ for which $(T-\lambda I)$ is not invertible, cf. e.g.~\cite{Rudin}.

We refer to~\cite{MT2009} for basic information on Markov chains. For any Markov kernel $P(x,dy)$ we denote by $\PM$ the associated  linear operator which acts on functions from $L^2(\pi)$, so that
\[\PM f(x)=\int_\stany f(y) P(x, d y)\;.\] 
 {Throughout the paper we will assume that considered Markov chains are $\pi$-irreducible.}
A Markov chain $(X_n)_{n\geq
0}$ with transition kernel $P(x,dy)$ and stationary distribution $\pi$ is called {\it geometrically ergodic} if  there exists $C(x)<\infty$ and $\vr<1$ such that $\Vert P^n(x,\cdot)-\pi(\cdot)\Vert_{TV}\leq C\vr^n$ holds for every $n$  { and a.e $x\in\stany$}.

Note that since $\PM\PiM=\PiM\PM=\PiM$ and $\PiM^n=\PiM$ we have the following fact. 
\begin{lemma}\label{lem:nnorm}
If $\PM$ is transition operator of a Markov chain with stationary distribution~$\pi$, then for every $n$ we have 
\[
\Vert \PM^n-\PiM\Vert=\Vert (\PM-\PiM)^n\Vert\leq\Vert \PM-\PiM\Vert^n\,.
\]
\end{lemma}
  We say that a Markov chain $(X_n)_{n\geq
0}$ with a transition operator $\PM$ and a stationary distribution $\pi$ has a {\it positive spectral gap} $(1-\vr)$ if and only if
\[\sup\{ |\lambda|\;:\; \lambda\in Spec(\PM-\PiM) \}=\vr<1\;,\]
where $Spec$ denotes the spectrum of operator $(\PM-\PiM)$ in $L^2(\pi)$. Then the spectral radius formula reads $\vr=\limsup_{n\to\infty}\|(\PM-\PiM)^n\|^\frac{1}{n}$, see \cite[Theorem~18.9]{Rudin}. Therefore, by Lemma~\ref{lem:nnorm}, we have \[\vr=\limsup_{n\to\infty}\|\PM^n-\PiM\|^\frac{1}{n}\,.\]
In turn, a Markov chain has a spectral gap if and only if there exists $n_0$ such that
\begin{equation}
    \label{MCSG}
\Vert \PM^{n_0}-\PiM\Vert<1\;.
\end{equation}
If in addition a Markov chain is reversible, then $\PM$ and $\PM -\PiM$ are self-adjoint operators and for any $n$ it holds that
$\Vert \PM^n-\PiM\Vert=\Vert \PM-\PiM\Vert^n$. So, in the reversible case    {$\vr=\Vert \PM-\PiM\Vert$}.\newline

\noindent \textbf{Geometry. }  For $i=1,\dots,d$ we define a subspace $M_i$ as subspace of functions $f\in L^2(\pi)$ which are constant with respect to $i$-th coordinate. We denote by $M:=\bigcap_{i=1}^d M_i$. Note that $M$ contains only
constant functions.

\medskip

It is easy to see that a small step  $\PM_i$, from the definitions of the Gibbs samplers,  is an orthogonal projection on $M_i$  for every $i$. Indeed, since operator $\PM_i(x,\cdot)$ is reversible with respect to $\pi$ so $\PM_i$ is a self-adjoint operator and $\PM_i^2=\PM_i$. Moreover, $\PiM$ is an orthogonal projection on $M$. We will use the following fact.
\begin{lemma}[\cite{badea2012rate}, Lemma 4.2] Let $f\in L^2(\pi)$. Suppose  $\PM_{\rm DSG}$ targeting $\pi$ is given by~\eqref{eqn_gibbs_kernel_DS} and $\PiM$ is given by~\eqref{PiM}.  Then for every $j\in [1,d]$, we have
\[\|\PM_{j-1}\dots \PM_1 f - \PM_j\dots \PM_1 f\|^2 \leq \|f-\PiM f\|^2 - \| \PM_{{\rm DSG}}f - \PiM f\|^2\,.\]
\label{lem:4.2}
\end{lemma}

In order to make use of some results of \cite{badea2012rate} we need to recall some objects introduced therein. We define the generalized Friedrichs angle $c(M_1,\dots,M_d)$ by
\begin{align}\label{c:def}
c(M_1,\dots,M_d)&:=\sup\left\{\frac{\sum_{i\neq j} \langle f_j,f_i\rangle}{(d-1)\sum_{i=1}^d \langle f_i,f_i\rangle}:\;\ f_i\in M_i\cap M^\bot,\ \sum_{i=1}^d \Vert f_i\Vert>0 \right\}\,.
\end{align}
In the case of two spaces this term was introduced in~\cite{friedrichs}. We point out that the notion of angle between Banach spaces is still attracting attention. See~\cite{oppenheim} for recent development.  {We also note that a quantity related to $c(M_1,\dots,M_d)$, called summation-based correlation coefficient,  was studied in the context of telescope property of spectral gap for Gibbs samplers by \cite{Qin2024}. Summation-based correlation coefficient is called a configuration constant by~\cite{badea2012rate} and they show it is equivalent (up to scaling and translation) to the generalized Friedrichs angle, see \cite[Proposition 3.6]{badea2012rate}.}

Let us note that from the very definition of norm and due to the fact that $M$ contains only constant functions we have
\[c(M_1,\dots,M_d)=\sup\left\{\frac{\sum_{i\neq j} \langle f_j,f_i\rangle}{(d-1)\sum_{i=1}^d \langle f_i,f_i\rangle}:\;\ f_i\in M_i,\ \pi f_i =0, \ \sum_{i=1}^d \pi f_i^2>0 \right\}\,.
\]
We also consider inclination  $\ell(M_1,\dots,M_d)$ as
\begin{equation}
    \label{l-def}
 \ell(M_1,\dots,M_d):=\inf_{f:\ \dist(f, M)=1}{\max_{1\leq i\leq d} \dist(f,M_i)}\;.
\end{equation}

We will use the following auxiliary fact resulting directly from \cite[Propositions~3.7 and~3.9]{badea2012rate}.
\begin{prop}\label{prop:3.7badea} Let $d\geq 2$. Suppose $M_1,\dots,M_d$ are closed subspaces of $L^2(\pi)$ with intersection $M=\bigcap_{i=1}^d M_d$. Assume further that $\PM_i$ for every $i$ is an orthogonal projection on $M_i$ and $\PiM$ is an orthogonal projection on $M$. Then\begin{equation}\label{eq:norm_random_scan}
\left\Vert \frac{1}{d} \sum_{i=1}^d \PM_i -\PiM\right\Vert=\frac{d-1}{d}\left[c(M_1,\dots,M_d)+\frac{1}{d-1}\right]
\end{equation}
and
\begin{equation}\label{eq:inclination}1-\frac{2d}{d-1}\ell(M_1,\dots,M_d)\leq c(M_1,\dots,M_d)\leq 1-\frac{1}{d-1}\ell^2(M_1,\dots,M_d)
\;.\end{equation}
\end{prop}

\section{Spectral gaps of Gibbs Samplers}\label{sec:gaps}

We establish precise sufficient conditions for the existence of the spectral gap for $\PM_{\rm DSG}$ and $\PM_{\rm RSG}$.

\begin{thm}\label{theo:equiv} For any stationary distribution $\pi$ on a Polish space $\stany = \stany_1 \times \dots \times \stany_d$ the following conditions are equivalent.
\begin{enumerate}[{\it (i)}]
 \item 
There exist positive weights $w_1,\dots,w_d$ as in~\eqref{weights} such that the random scan Gibbs sampler $\PM_{\rm RSG} =\sum_{i=1}^d w_i \PM_i,$ has an $L^2(\pi)$ spectral gap.
 \item 
For every choice of weights  $w_1,\dots,w_d$ satisfying~\eqref{weights} the  random scan Gibbs sampler $\PM_{\rm RSG} =\sum_{i=1}^d w_i \PM_i$ has an $L^2(\pi)$ spectral gap.
 \item 
There exists a permutation $\sigma$ such that the deterministic scan Gibbs sampler $\PM_{\rm DSG} =\PM_{\sigma(1)}\PM_{\sigma(2)}\cdots \PM_{\sigma(d)}$ has an $L^2(\pi)$ spectral gap.
 \item 
For every permutation $\sigma$ the deterministic scan Gibbs sampler $\PM_{\rm DSG} =\PM_{\sigma(1)}\PM_{\sigma(2)}\cdots \PM_{\sigma(d)}$ has an $L^2(\pi)$ spectral gap.
 \item 
There exists a permutation $\sigma$ such that $\left\Vert \PM_{\sigma(1)}\PM_{\sigma(2)}\cdots \PM_{\sigma(d)}-\PiM\right\Vert<1$.
 \item 
For every permutation $\sigma$ it holds that $\left\Vert \PM_{\sigma(1)}\PM_{\sigma(2)}\cdots \PM_{\sigma(d)}-\PiM\right\Vert<1$. 
 \end{enumerate}
\end{thm}
Conditions {\it (v)} and {\it (vi)} are in general stronger than {\it (iii)} and {\it (iv)}, respectively. However, we show that the reverse implication is also true.
 In addition, since $\|T^*T\|=\|T^2\|$ for any operator $T$, conditions {\it (v)}, {\it (vi)} are equivalent to  {\[\left\Vert \PM_{\sigma(1)}\PM_{\sigma(2)}\cdots \PM_{\sigma(d)}\PM_{\sigma(d-1)}\cdots\PM_{\sigma(2)}\PM_{\sigma(1)}-\PiM\right\Vert<1\] for some permutation $\sigma$.}

\smallskip

Let us note that conditions {\it (v)} and {\it (vi)} are used to obtain Hoeffding's inequality for Markov chains \citep{miasojedowhoeffding}, see Theorem \ref{thm:Hoeffding} in Section \ref{sec:CTG}. 
\newline

In the proof of Theorem~\ref{theo:equiv} we employ the following lemmas.

\begin{lemma}\label{lem:norm}
For any Markov chain with transition operator $\PM$ and stationary distribution $\pi$ we have
\[\Vert \PM -\PiM\Vert\leq 1\,.\]
 
\end{lemma}
\begin{proof}
 Let us denote $f_0=f-\pi f$. Simple calculations shows that $\Vert f_0\Vert\leq \Vert f\Vert $. Indeed,
 \[\Vert f_0\Vert=\pi f_0^2=\pi((f-\pi(f))^2)=\pi(f^2)-(\pi(f))^2\leq\pi(f^2)=\Vert f\Vert\,.\]
 For any $f\in L^2(\pi)$ by Jensen inequality and by the definition of stationary distribution we have
 \[\left\| \PM f-\PiM f\right\|=\left\Vert \PM f_0\right\Vert=\int_\stany (\PM f_0(x))^2\pi(x) dx\leq\int_\stany \PM f_0^2(x)\pi(x) dx= \pi(f_0^2)\leq \Vert f\Vert\,. \]
\end{proof}
 {The next lemma can be also deduced from \cite[Proposition 2]{JRR}, but we give a simple proof for the sake of completeness.}
\begin{lemma}\label{lem:weighst}If there exist $w_1,\dots,w_d$ satisfying~\eqref{weights} and such that 
\[\left\Vert \sum_{i=1}^d w_i \PM_i -\PiM\right\Vert<1\,,\]
then for all weights $\tilde{w}_1,\dots,\tilde{w}_d$ satisfying~\eqref{weights}, we have
\[\left\Vert \sum_{i=1}^d \tilde w_i \PM_i -\PiM\right\Vert<1\,.\]
\end{lemma}
\begin{proof}For any arbitrarily chosen $\tilde{w}_1,\dots,\tilde{w}_d$, 
 let $\alpha=\min_{1\leq i \leq d} \frac{\tilde w_i}{w_i}$. Obviously  $\alpha\leq 1$. For $i=1,\dots,d$ we define $\hat w_i = \frac{\tilde w_i -\alpha w_i}{1-\alpha}$. By definition
 $\hat w_i\geq0$ and $\sum_{i=1}^d \hat w_i=1$. Let $\hat \PM = \sum_{i=1}^d \hat w_i \PM_i$. Clearly, $\hat \PM$ is a Markov transition operator with stationary distribution $\pi$.
 We have
 \[\left\| \sum_{i=1}^d \tilde w_i \PM_i -\PiM\right\|=\left\| \alpha \left[\sum_{i=1}^d w_i \PM_i-\PiM\right]+(1-\alpha) \left[\hat \PM-\PiM\right]\right\|\,.\]
 Applying the triangle inequality we obtain
 \begin{equation}\label{eq:triangle}
 \left\| \sum_{i=1}^d \tilde w_i \PM_i -\PiM\right\|\leq \alpha \left\|\sum_{i=1}^d w_i \PM_i-\PiM\right\|+(1-\alpha) \left\|\hat\PM-\PiM\right\|\,. 
 \end{equation}

 Finally, by the assumption and by Lemma~\ref{lem:norm} we get
  \[\left\| \sum_{i=1}^d \tilde w_i \PM_i -\PiM\right\|<\alpha+(1-\alpha)=1\,.\]
\end{proof}

\begin{lemma}\label{lem:norm_determ} Let $d\geq 2$. Suppose $M_1,\dots,M_d$ are closed subspaces of $L^2(\pi)$ with intersection $M=\bigcap_{i=1}^d M_d$. Assume further that $\PM_i$ for every $i$ is an orthogonal projection on $M_i$ and $\PiM$ is an orthogonal projection on $M$. Then for every permutation $\sigma$ we have
\begin{equation}\label{eq:norm_deterministic}
\left\Vert \PM_{\sigma(1)}\PM_{\sigma(2)}\cdots \PM_{\sigma(d)}-\PiM\right\Vert<\sqrt{1-\frac{ {\ell^2} }{d^2}}\,,
\end{equation}
 {where $\ell:=\ell(M_1,\dots,M_d)$ is defined in~\eqref{l-def}}.
\end{lemma}
\begin{proof}We follow the ideas of the proof of \cite[Theorem~4.1]{badea2012rate}. By Lemma~\ref{lem:4.2} for every $j\in [1,d]$ we can estimate
\begin{align*}
\left(\dist(f, M_j )\right)^2 &=\|f-\PM_j f\|^2\leq \|f - \PM_j\dots \PM_1 f\|^2\\
&\leq (\|f - \PM_1 f\| + \|\PM_1 f - \PM_2\PM_1f\|+\cdots+ \|\PM_{j-1}\dots \PM_1 f - \PM_j\dots \PM_1 f\|)^2\\
&\leq j\left(\|f - \PM_1 f\|^2 + \|\PM_1 f - \PM_2\PM_1f\|^2+\cdots+ \|\PM_{j-1}\dots \PM_1 f - \PM_j\dots \PM_1 f\|^2\right)\\
&\leq  j^2 \left(\|f - \PiM f\|^2 - \|\PM_{\rm DSG}f - \PiM f\|^2\right)\\
&\leq  d^2\left(\|f - \PiM f\|^2 - \|\PM_{\rm DSG}f - \PiM f\|^2\right).
\end{align*}
 {Therefore}, since $\|f - \PiM f\|^2=\dist (f,M)$ by the very definition of $\ell$ we see that
\begin{align*}
 \ell ^2 \|f - \PiM f\|^2&\leq \left(\max_{j\in [1,d]}\dist(f, M_j )\right)^2 \leq d^2(\|f - \PiM f\|^2 - \|\PM_{\rm DSG}f - \PiM f\|^2).
\end{align*}
Consequently, taking into account that $\|f-\PiM f\|^2\leq \|f\|^2$, we get
\[\|\PM_{\rm DSG}f - \PiM f\|^2 \leq \left(1- \frac{ \ell ^2}{d^2}\right)\|f-\PiM f\|^2 \leq \left(1- \frac{ \ell ^2}{d^2}\right)\|f\|^2 \,. \]
Above, the order of projections does not matter, so we infer that for any permutation $\sigma$ of set $\{1,\dots,d\}$ we have~\eqref{eq:norm_deterministic}. 
\end{proof}

Let us prove the equivalence of sufficient conditions for the existence of the spectral gap for both types of Gibbs samplers.

\begin{proof}[Proof of Theorem~\ref{theo:equiv}] As the implications $(ii)\Rightarrow(i)$, $(iv)\Rightarrow(iii)$, and $(vi)\Rightarrow(v)$ are obvious, let us concentrate on the remaining ones. In this proof we denote $\ell:=\ell(M_1,\dots,M_d)$.
\begin{itemize}
    \item[$(i)\Rightarrow(ii)$] 
 For any weights $w_1,\dots,w_d$ random scan Gibbs sampler is reversible. Hence, the spectral gap property is equivalent to $\left\Vert \sum_{i=1}^d w_i \PM_i -\PiM\right\Vert<1$ and $(ii)$ results from Lemma~\ref{lem:weighst}.
 
    \item[$(i)\Rightarrow(vi)$] From $(i)$ and Lemma~\ref{lem:weighst} we get that  $\left\Vert \frac{1}{d} \sum_{i=1}^d \PM_i -\PiM\right\Vert<1$. On the other hand, the assumptions of Proposition~\ref{prop:3.7badea} are satisfied and we have~\eqref{eq:norm_random_scan}.  Therefore 
$c(M_1,\dots,M_d)<1$. Due to~\eqref{eq:inclination} we obtain that $\ell>0$. By Lemma~\ref{lem:norm_determ}  
we have \[
\left\Vert \PM_{\sigma(1)}\PM_{\sigma(2)}\cdots \PM_{\sigma(d)}-\PiM\right\Vert<\sqrt{1-\frac{\ell^2 }{d^2}}<1\,\]
and $(vi)$ follows.

    \item[$(vi)\Rightarrow(iv)$] We take any fixed permutation $\sigma$. Then it holds that \[\left\Vert \PM_{\sigma(1)}\PM_{\sigma(2)}\cdots \PM_{\sigma(d)}-\PiM\right\Vert<1\] and $(iv)$ results from~\eqref{MCSG}.
    
    \item[$(v)\Rightarrow(iii)$] It is a special case of 
$(vi)\Rightarrow(iv)$.

 \item[$(iii)\Rightarrow(i)$]  {Follows directly from \cite[Proposition 3.2]{roberts1997geometric}}
\end{itemize} 
\end{proof}

We are now in a position to infer our main result.
\begin{proof}[Proof of Theorem~\ref{theo:main}]
We notice that the desired dichotomy is already shown in the proof of Theorem~\ref{theo:equiv}. In fact, if any of the Gibbs samplers has a positive $L^2(\pi)$ spectral gap, then property {\it (i)} or {\it (iii)} from Theorem~\ref{theo:equiv} holds true. Then, by {\it (ii)} and {\it (iv)}, all of the Gibbs samplers have a positive $L^2(\pi)$ spectral gap. Otherwise none of the Gibbs samplers has a positive $L^2(\pi)$ spectral gap.
\end{proof}

\begin{proof}[Proof of Corollary~\ref{cor:main}] For any  {$\pi$-irreducible} Markov chain the spectral gap property implies GE, \cite[Theorem~1.3]{kontoyiannis2012geometric}. Moreover,  for reversible chains geometric ergodicity is equivalent to spectral gap property, \cite[Theorem~2.1]{roberts1997geometric} or \cite[Proposition~1.2]{kontoyiannis2012geometric}.
 Since the random scan Gibbs sampler is reversible, in the view of equivalent conditions of Theorem~\ref{theo:equiv} we infer that all Gibbs samplers have positive spectral gaps. 
 \end{proof}

Using the proof of Theorem~\ref{theo:equiv} we can estimate the spectral gap of the   Gibbs samplers. Let us point out that the estimate for the random scan Gibbs sampler is sharp.
\begin{cor}\label{cor:bounds} { Let $c:=c(M_1,\dots,M_d)$ be given by~\eqref{c:def}.} The following estimates on the spectral gaps hold true. \begin{enumerate}[(i)]
    \item 
If $\alpha:=d\min_{1\leq i\leq d} w_i$, for the random scan Gibbs sampler $\PM_{\rm RSG} =\sum_{i=1}^d w_i \PM_i$ with any weights $w_1,\dots,w_d$ satisfying~\eqref{weights} the following bound holds true 
 \begin{equation}\label{cor1}
  \left\|\PM_{\rm RSG} -\PiM\right\| \leq \tfrac{d-1}{d}\alpha\left( {c} +\tfrac{1}{d-1}\right)+1-\alpha \,.
 \end{equation} 
Moreover,  equality holds if $ w_i=\frac{1}{d}$ for every $i=1,\dots, d$ and
\begin{equation}\label{cor1a}
\tfrac{1}{d}\leq  \left\|\PM_{\rm RSG} -\PiM\right\| \,.
 \end{equation}

\item For any permutation $\sigma$ of set $\{1,\dots,d\}$ for the deterministic scan Gibbs sampler $\PM_{\rm DSG} =\PM_{\sigma(1)}\PM_{\sigma(2)}\cdots \PM_{\sigma(d)}$ we have
 \begin{equation}
\label{cor2}
 \Vert \PM_{\rm DSG} -\PiM\Vert\leq\sqrt{1- {\frac{(d-1)^2}{4d^4}(1-c )^2}}\,.
 \end{equation}
\end{enumerate}
\end{cor}

\begin{proof}
 \begin{enumerate}[\it (i)]
    \item We notice that
 \begin{align*}
     \left\| \sum_{i=1}^d w_i \PM_i -\PiM\right\|&=\left\| \alpha \left[\sum_{i=1}^d \frac{1}{d} \PM_i-\PiM\right]+(1-\alpha) \left[\sum_{i=1}^d \frac{w_i-\frac{\alpha}{d}}{1-\alpha} \PM_i-\PiM\right]\right\|\\
&\leq \alpha \left\|\sum_{i=1}^d \frac{1}{d} \PM_i-\PiM\right\|+1-\alpha\,,
 \end{align*}
 where we used the triangle inequality and Lemma~\ref{lem:norm}. By applying~\eqref{eq:norm_random_scan} we obtain the desired inequality.
 
If $w_i=\frac 1d$  for every $i=1,\dots, d$, then \eqref{eq:norm_random_scan} implies that we deal with the equality in~\eqref{cor1}. To justify~\eqref{cor1a} it suffices to note that by the very definition  {$c=c(M_1,\dots,M_d)
\geq 0$.}

\item Inequality~\eqref{cor2} is a consequence of \eqref{eq:norm_deterministic} combined with \eqref{eq:inclination}. \end{enumerate}
\end{proof}

\begin{cor}\label{cor:poly}
 {
 In terms of the notion of rapid and torpid mixing \citep{torpid,rapid}, if a random scan mixes rapidly, so does any deterministic scan. Equivalently, if a deterministic scan mixes torpidly, so does any random scan.}
\end{cor}
\begin{proof}
 {
If a random scan mixes rapidly, its spectral gap  $1-\varrho_{{\rm RSG}}(d)$ decays polynomially with dimension $d$, that is $\lim_{d\to\infty}\frac{1-\varrho_{{\rm RSG}}(d)}{d^{-\beta}}>0$ for some $\beta>0$. Therefore, for large enough $d$ we have $\varrho_{{\rm RSG}}(d)\leq 1 -\gamma d^{-\beta}$ for some $\gamma>0$. We pick such $\gamma>0$. Hence, by \eqref{eq:norm_random_scan} we get
\[
c(M_1,\dots,M_d)\leq 1 - \gamma d^{-\beta}\,.
\]
 This together with~\eqref{cor2} the above implies that   \[\varrho_{{\rm DSG}}^2(d)\leq 1- \tfrac{(d-1)^2}{4d^2}\gamma^2  d^{-2\beta-2}\leq 1-\tfrac{\gamma^2}{16} d^{-2\beta-2}\,. 
\]
Hence
\[
1-\varrho_{{\rm DSG}}(d)\geq \tfrac{1}{2}(1-\varrho_{{\rm DSG}}^2(d))\geq \tfrac{\gamma^2}{32}d^{-2\beta-2}\,,
\]
which means that a deterministic scan also mixes rapidly.  
}
\end{proof}
We pose the question whether the reverse implication also holds true.
\begin{open}\label{open:poly}
Is it possible to prove that the decay of $\varrho_{{\rm RSG}}(d)$ is a polynomial function of the dimension $d$, knowing that $\varrho_{{\rm DSG}}(d)$ decays with a polynomial rate with dimension $d$?
\end{open}

\section{Central Limit Theorem and Hoeffding inequality for Gibbs Samplers}\label{sec:CTG}

As a direct consequence of Theorem~\ref{theo:equiv} we prove that the following version of Central Limit Theorem and Hoeffding inequality hold for both  {versions of} Gibbs samplers.

\begin{thm}\label{thm:clt}
Let  a Markov chain $(X_n)_{n\geq
0}$ with a stationary distribution $\pi$ be generated by  the deterministic scan Gibbs sampler $\PM_{\rm DSG}$ with a positive spectral gap or by the random scan Gibbs sampler $\PM_{\rm RSG}$ with a positive spectral gap.  {Assume further that  $(X_n)_{n\geq
0}$ is Harris recurrent}.
Then for any function $f$ with $\pi( f^2)<\infty$  { and any initial distriubtion} we have
\[\frac{1}{\sqrt{n}}\left[\sum_{i=0}^{n-1}\left( f(X_i)-\pi f\right)\right] \overset{D}{\rightarrow}\mathcal{N}(0,\sigma^2(f))\,,\]
 where asymptotic variance $\sigma^2(f)$ is bounded by
 \[\sigma^2(f)\leq\frac{1+\vr}{1-\vr}\pi(f-\pi f)^2\,,\]
 with  $(1-\vr)$ being the $L^2(\pi)$ spectral gap of the considered Gibbs sampler.\\
 Moreover,  for the random scan Gibbs sampler $\PM_{\rm RSG} =\sum_{i=1}^d w_i \PM_i$ with any weights $w_1,\dots,w_d$ satisfying~\eqref{weights} we have $ \vr=   \left\|\PM_{\rm RSG} -\PiM\right\|$, whereas for the deterministic scan Gibbs sampler $\PM_{\rm DSG} =\PM_{\sigma(1)}\PM_{\sigma(2)}\cdots \PM_{\sigma(d)}$ with any permutation $\sigma$ of set $\{1,\dots,d\}$ we have
 $\vr\leq
 \Vert \PM_{\rm DSG} -\PiM\Vert$.  {In both cases, $\vr$ and $\sigma^2(f)$ enjoy explicit bounds in terms of dimension $d$ and generalized Friedrichs angle $c(M_1,\dots,M_d)$.}
 \end{thm}
 
\begin{proof} Let $\PM$ be a transition operator for $(X_n)_{n\geq
0}$. We will employ a general result for geometrically ergodic Markov chains, i.e. \cite[Theorem 17.4.4]{MT2009}. By Theorem~\ref{theo:equiv} there exists $\vr<1$ such that $\Vert \PM -\PiM\Vert\leq\vr$ and $\Vert \PM^i-\PiM\Vert\leq \vr^i$. We define $\hat f = \sum_{i=0}^\infty (\PM^i f-\PiM f)$. By the triangle inequality we infer that
 \[\Vert \hat f \Vert\leq\sum_{i=0}^\infty \Vert \PM^i f -\PiM f\Vert\leq \sqrt{\PiM f^2} \sum_{i=0}^\infty \Vert \PM^i-\PiM\Vert\leq\sqrt{\PiM f^2}\sum_{i=0}^\infty \vr^i<\infty\,.\]
Moreover, $\hat f$ satisfies the so-called Poisson equation
 \[ \hat f - \PM \hat f =f -\PiM f\;.\]
Therefore,  assumptions of  \cite[Theorem 17.4.4]{MT2009} are verified and we get that 
\[\frac{1}{\sqrt{n}}\left[\sum_{i=0}^{n-1}\left( f(X_i)-\pi f\right)\right] \overset{D}{\rightarrow}\mathcal{N}(0,\sigma^2(f))\,,\]
where $\sigma^2(f)=\pi({\hat f}^2 -(\PM\hat f)^2)$. Moreover, we have
\[\sigma^2(f)=\pi(f-\pi f)^2+2\sum_{i=1}^\infty \langle f-\pi f,(\PM-\PiM)(f-\pi f)\rangle\,.\]
By the Cauchy--Schwartz inequality and since $\Vert \PM^i-\PiM\Vert\leq \vr^i$ we get
\[\sigma^2(f)\leq \pi(f-\pi f)^2\left[1+2\sum_{i=1}^\infty \vr^i\right]=\frac{1+\vr}{1-\vr}\pi(f-\pi f)^2\,.
\]
 {The bounds for $\vr$ for $\PM_{\rm RSG}$ and $\PM_{\rm DSG}$ result from \eqref{cor1} and \eqref{cor2}, respectively.} 
 \end{proof}
 {Another consequence of our analysis and Theorem~\ref{theo:equiv} is that both versions of the Gibbs sampler satisfy the following exponential inequality.}
\begin{thm}\label{thm:Hoeffding}
 Let  a Markov chain $(X_n)_{n\geq
0}$ and stationary distribution $\pi$ be generated by  the deterministic scan Gibbs sampler $\PM_{\rm DSG}$ with a positive spectral gap or by the random scan Gibbs sampler $\PM_{\rm RSG}$ with a positive spectral gap. Let $f$ be a function   {valued in} $[0,1]$ and denote $\mu :=\pi(f)$. Then for any $\varepsilon$ such that $\mu+\varepsilon \leq 1$, any $n\geq 1$, and any initial measure  {$\nu\ll\pi$ with $\tfrac{d\nu}{d\pi}\in L^2(\pi)$}, it holds that
\[
\Pr\left(\sum_{i=1}^n f(X_i)\geq n(\mu+\varepsilon)\right)\leq
\left\Vert \tfrac{d\nu}{d\pi}\right\Vert \exp\left(-\tfrac{1-\vr}{1+\vr}n\varepsilon^2\right)\]
with  $(1-\vr)$ being the $L^2(\pi)$-spectral gap of the considered Gibbs sampler.\\
 Moreover,  for the random scan Gibbs sampler $\PM_{\rm RSG} =\sum_{i=1}^d w_i \PM_i$ with any weights $w_1,\dots,w_d$ satisfying~\eqref{weights} we have $ \vr=   \left\|\PM_{\rm RSG} -\PiM\right\|$, whereas for the deterministic scan Gibbs sampler $\PM_{\rm DSG} =\PM_{\sigma(1)}\PM_{\sigma(2)}\cdots \PM_{\sigma(d)}$ with any permutation $\sigma$ of set $\{1,\dots,d\}$ we have
 $\vr\leq
 \Vert \PM_{\rm DSG} -\PiM\Vert$.   {In both cases, $\vr$ enjoys an explicit bound in terms of dimension $d$ and generalized Friedrichs angle $c(M_1,\dots,M_d)$.}
\end{thm} 
 \begin{proof}
 Combining Theorem~\ref{theo:equiv} with \cite[Theorem 1.1]{miasojedowhoeffding} we get required inequality.  {The bounds for $\vr$ for $\PM_{\rm RSG}$ and $\PM_{\rm DSG}$ result from \eqref{cor1} and \eqref{cor2}, respectively.} 
 \end{proof}

\section{On Open problem~\ref{op}}\label{sec:op}

In this section we present the conditions equivalent to the positive answer to Open problem~\ref{op} and give an example suggesting that the answer might be negative.\newline

\begin{remark} 
In order to show that from GE of $\PM_{\rm DSG}$ the hypothesis of Theorem~\ref{theo:main} follows it suffices to prove that GE of $\PM_{\rm DSG}$ implies GE of $\PM_{\rm RSG}$ (or, equivalently, any of the spectral gap properties from Theorem~\ref{theo:equiv}).   GE of $\PM_{\rm DSG}$ is equivalent to each of the following assertions, see \cite[Proposition~2.1]{roberts1997geometric} and \cite[Theorem~15.0.2]{MT2009}.
\begin{enumerate}[$(i)$]
    \item {\it (exponential moment of return time)} There exists some small measurable set $C$ and $\tau$ -- return time to set $C$, and $b > 1$ such that $\sup_{x\in C} \mathbb{E}(b^{\tau}|X_0=x) < \infty$.
    \item {\it (drift condition)} There exists a measurable function $V\geq 1$, a small set $C$, $\lambda<1$, and $K<\infty$, such that $\PM_{\rm DSG}V(x)\leq \lambda V(x)+K\mathds{1}_C(x)$.
    \item {\it (spectral gap in $L_V$)} For a measurable function $V\geq 1$ and a space endowed with a norm $\|f\|_{L_V}:=\sup_{x\in\stany} \frac{|f(x)|}{V(x)}$ the operator  $(\PM_{\rm DSG}-\PiM)$ has a spectral gap in $L_V$.
\end{enumerate} 
A set $C$ is called small if $\pi(C)>0$ and there exist $n,$ $\beta>0$, and a probability measure $\mu$, such that for every measurable set $A$ we have $P_{\rm DSG}^n(x,A)\geq \mathds{1}_C(x)\beta\mu(A).$
\end{remark}

\smallskip

Let us illustrate with an example why we suspect that it might not be possible to infer from GE of $\PM_{\rm DSG}$   {the existence of its positive spectral gap. 
}   {Let us observe that for a~Markov chain with a transition operator $\PM$, its spectral gap is positive if and only if  its additive reversibilization $\frac{1}{2}(\PM+\PM^*)$ has a positive spectral gap. That can happen if  the operator $\frac{1}{2}(\PM+\PM^*)$  is geometrically ergodic.}   {In fact, we construct a~Markov chain, such that  $\PM$ and $\PM^*$ are geometrically ergodic, but  $\frac{1}{2}(\PM+\PM^*)$ is not. Therefore, to obtain a possible positive answer to Open Problem~\ref{op}, one must rely on very special properties of Gibbs samplers.}
\begin{example} {Let $\mathbb{Z}^+=\{0,1,2,\dots\}$.}  Let us define the Markov chain $\{X_n\}_{n\geq0}$ on state space $\mathbb{Z}^+\times\mathbb{Z}^+$ by its transition matrix $\PM$ as follows
 \begin{align*}
  P((0,0),(n,n)) = p(n)&\quad \text{for all } n\in \mathbb{Z}^+,\\
  P((n,k),(n,k-1))=1&\quad \text{for all } n\in \mathbb{Z}^+,\ k=2,\dots n\,,\\
  P((n,1),(0,0))=1&\quad \text{for all } n\in \mathbb{Z}^+,\\
  P((k,l),(i,j))=0 &\quad \text{otherwise}\,,
 \end{align*}
where $p(n)$ is arbitrarily chosen distribution with $p(n)>0$ for all $n$ and with finite exponential moments. Clearly, Markov chain $\{X_n\}_{n\geq0}$ is aperiodic and irreducible.  
We define the stopping time $\tau=\inf\{n:\, n>0,\, X_n=(0,0)\}$. By \cite[Theorem 15.0.2]{MT2009} to establish GE it is enough to show that for some $b>1$ we have an exponential moment $\mathbb{E}(b^\tau|X_0=(0,0))<\infty$. But if the chain starts at $(0,0)$ the distribution of $\tau$ is given by $\mathbb{P}(\tau=n)=p(n-1)$ for $n\geq1$. 
Therefore from assumption on $p(n)$ we get that the exponential moment of $\tau$ is finite and further that the chain $\{X_n\}_{n\geq0}$ is geometrically ergodic.
From  \cite[Theorem 10.0.1]{MT2009} we deduce that the stationary distribution $\pi$  {is independent of the second variable.} 
Hence the adjoint operator $\PM^*$
defined by equation 
\[\pi(k,l)P^*((k,l),(i,j))=\pi(i,j)P((i,j),(k,l))\]
is given by
 \begin{align*}
  P^*((0,0),(n,1)) = p(n)&\quad \text{for all } n\in \mathbb{Z}^+\,,\\
  P^*((n,k-1),(n,k))=1&\quad \text{for all } n\in \mathbb{Z}^+,\ k=2,\dots n\,,\\
  P^*((n,n),(0,0))=1&\quad \text{for all } n\in \mathbb{Z}^+\,,\\
  P^*((k,l),(i,j))=0 &\quad \text{otherwise}\,,
 \end{align*}
The chain generated according to $\PM^*$, by analogous arguments, is also geometrically  ergodic. We will prove though that its additive reversibilisation $\mathsf{K}:=\frac{1}{2}(\PM+\PM^*)$ is not geometrically ergodic. It suffices to show that the conductance for this chain is $0$, because then by Cheeger's inequality~\cite[Theorem~2.1]{lawler1988bounds} we will get that chain does not have a spectral gap. Consequently, by reversibility of $\mathsf{K}$, we obtain
that $\mathsf{K}$ is not geometrically ergodic. The conductance of $\mathsf{K}$ is defined as
\[\varkappa:= \inf_{A\,:\,\pi(A)>0}\frac{\sum_{x\in A}\pi(x)K(x,A^c) }{\pi(A)\pi(A^c)}\,,\]
where $A^c$ is the complement of $A$. For any $n>0$ we define a set $A_n=\{(n,k)\,:\, k=1,\dots,n\}$. Observe that $K((n,k),A_n^c)=0$ for $k=2,\dots,n-1$ and $K((n,k),A_n^c)=\frac{1}{2}$ for
$k\in\{1,n\}$. Since $\pi(A_n^c)>\pi(0,0)>0$  and $\pi(n,k)=\pi(n,j)=\delta$,  { for some $\delta\geq 0$} we get
\[\varkappa\leq \frac{\sum_{x\in A_n}\pi(x)K(x,A^c_n) }{\pi(A_n)\pi(A^c_n)}\leq \frac{\delta }{n\delta\pi(0,0)}=\frac{1}{n\pi(0,0)}\xrightarrow[n\to\infty]{}0\,.\]
Since $n$ can be arbitrarily large, we infer that the conductance $\varkappa=0$, and further chain $\mathsf{K}$ is not geometrically ergodic.
\end{example}

\section{Acknowledgements}
We thank anonymous referees for their very useful comments that improved the paper. The research of IC is supported by NCN grant 2019/34/E/ST1/00120, the research of BM is supported by NCN grant 2018/31/B/ST1/00253. K{\L} has been supported by the Royal Society through the Royal Society University Research Fellowship.
\bibliographystyle{imsart-nameyear}

\bibliography{refs}

\end{document}